\documentclass[a4paper, english]{article}
\usepackage[utf8]{inputenc}
\usepackage{amsmath,amsthm,amssymb,mathtools}
\usepackage{algorithm,algorithmic}
\usepackage{microtype}
\usepackage{cite}
\usepackage{fullpage}
\usepackage{enumitem}
\usepackage{xcolor}
\usepackage{authblk}
\usepackage{booktabs}
\usepackage[small]{caption}

\usepackage{stmaryrd}
\usepackage{xspace}

\title{Distributed Reconfiguration of Maximal Independent Sets}

\author[1]{Keren Censor-Hillel}
\author[2]{Mikael Rabie}
\affil[1]{Department of Computer Science, Technion, Israel, ckeren@cs.technion.ac.il}
\affil[2]{IRIF, Universit\'e Paris Diderot, France, mikael.rabie@irif.fr}

\usepackage{tikz}
\usetikzlibrary{calc}
\usetikzlibrary{chains,scopes,positioning}

\newtheorem{theorem}{Theorem}
\newtheorem{corollary}{Corollary}
\newtheorem{lemma}{Lemma}

\newtheorem{definition}{Definition}

\newenvironment{theorem-repeat}[1]{\begin{trivlist}
		\item[\hspace{\labelsep}{\bf\noindent Theorem \ref{#1}. }]\em }%
	{\end{trivlist}}
	
\newenvironment{lemma-repeat}[1]{\begin{trivlist}
	\item[\hspace{\labelsep}{\bf\noindent Lemma \ref{#1}. }]\em }%
{\end{trivlist}}

\usepackage{cite}
\begin{document}

\maketitle
\setcounter{footnote}{0}

\begin{abstract}
In this paper, we investigate a \emph{distributed maximal independent set (MIS) reconfiguration problem}, in which there are two maximal independent sets for which every node is given its membership status, and the nodes need to communicate with their neighbors in order to find a \emph{reconfiguration schedule} that switches from the first MIS to the second. Such a schedule is a list of \emph{independent sets} that is restricted by forbidding two neighbors to change their membership status at the same step. In addition, these independent sets should provide some covering guarantee.

We show that obtaining an actual MIS (and even a 3-dominating set) in each intermediate step is impossible. However, we provide efficient solutions when the intermediate sets are only required to be independent and 4-dominating, which is almost always possible, as we fully characterize.

Consequently, our goal is to pin down the tradeoff between the possible length of the schedule and the number of communication rounds.
We prove that a constant length schedule can be found in $O(\texttt{MIS}+\texttt{R32})$ rounds, where $\texttt{MIS}$ is the complexity of finding an MIS in a worst-case graph and $\texttt{R32}$ is the complexity of finding a $(3,2)$-ruling set. For bounded degree graphs, this is $O(\log^*n)$ rounds and we show that it is necessary. On the other extreme, we show that with a constant number of rounds we can find a linear length schedule.
\end{abstract}

\section{Introduction}
Consider a distributed setting in which each node of a network receives an input from a higher-level application which tells it whether it is \emph{selected} or not, such that the set of selected nodes is a maximal independent set (MIS), which we will denote by $\alpha$. The reason that the application requires an MIS is because it needs the set of selected nodes to dominate all nodes for the sake of, say, monitoring the network, but without having violations of two neighbors being in the set, because they may cause conflicting actions. Now, because of changes in the network traffic, the energy consumption, or any one of various conditions that may change, the application needs to change the selected set of nodes. Once a new input MIS, denoted by $\beta$, is given to the nodes by the application, the nodes need to \emph{reconfigure} their states to that set while never sacrificing the safety condition of independence. In fact, for compatibility reasons, neighboring nodes cannot change their membership in the set at the same time, so a \emph{sequence} of changes is needed for converging into the new MIS. We call such a sequence a \emph{reconfiguration schedule}.

The length of the schedule is clearly a measure that is required to be minimized. Hence, an extreme solution would be to have all nodes declare themselves as unselected, and then the new set of nodes declare that they are selected. However, this very fast approach suffers from loosing the domination property throughout the reconfiguration schedule. Thus, the structure of the network must be taken into account, but since the topology is unknown, finding a schedule that maintains a good covering at all times necessitates that the nodes communicate. This brings another measure of complexity into question, which is the number of communication rounds that are needed in order to find a short schedule. Our goal in this work is to study the tradeoff between the possible length of the schedule and the number of communication rounds needed for finding it.

Unfortunately, as we show, it is not always possible to find schedules where each set is an MIS. This impossibility holds even if we relax the condition of domination and require only independent 3-dominating sets. Even when 3-domination \emph{is} possible, it may be extremely inefficient (Section~\ref{sec:impossibilities}).
\newcommand{\PropThreeDom}
{
Requiring $3$-domination for intermediate steps is costly:
\begin{enumerate}
\item There exists a class of inputs $G=(V,E)$ with two MIS $\alpha$ and $\beta$ such that there is no reconfiguration schedule with 3-dominating intermediate steps.
\item There exists a class of inputs $G=(V,E)$ with two MIS $\alpha$ and $\beta$ such that any reconfiguration schedule is of length $\Omega(n)$ and needs $\Theta(n)$ rounds to be found, if intermediate steps must be 3-dominating.
\end{enumerate}
}
\begin{theorem}
\label{prop:3dom}
\PropThreeDom
\end{theorem}

However, we prove that independence and 4-domination can indeed be obtained. Our main result is the following (Section~\ref{sec:const-len}).
\begin{theorem}\label{th:main}\textbf{(informal)}
For any graph $G=(V,E)$ of diameter greater than 3 and any input of two MIS $\alpha,\beta$, there exists a reconfiguration schedule of constant length $28$, with independent 4-dominating intermediate steps. Moreover, such a schedule can be found in $O(\texttt{MIS}+\texttt{R32})$ rounds, where $\texttt{MIS}$ is the complexity of finding an MIS on a worst-case graph and $\texttt{R32}$ is the complexity of finding a $(3,2)$-ruling set on a worst-case graph.
\end{theorem}

Obtaining the above theorem turns out to be an involved task. Our key ingredients are the following. We prove that graphs with a not-too-small diameter always admit a schedule of reconfiguration steps from a given maximal independent set to another. Moreover, full knowledge of the topology of the graph is not necessary in order to be able to locally add an element to the set after having removed its neighbors (to avoid dependence). Rather, only local manipulations are needed for doing so.

The currently known complexities that give $O(\texttt{MIS}+\texttt{R32})$ are discussed in the related work part. Here, we draw attention to
the fact that an immediate corollary of Theorem~\ref{th:main} is that for graphs of bounded degree we can compute the constant length schedule within $O(\log^*n)$ rounds. Further, we show that this is a lower bound by reducing the problem of finding an MIS on a path to obtaining a constant-length schedule for MIS reconfiguration. The following theorem actually holds even if one requires only $d$-domination, for some constant $d\geq 4$ (Section~\ref{sec:impossibilities}).

\newcommand{\Proplogstar}
{
For any fixed $k$, there exists a class of $k$-regular inputs $G=(V,E)$ with two MIS $\alpha$ and $\beta$ such that any reconfiguration schedule of constant length with 4-domination needs $\Theta(\log^*n)$ rounds to be found.
}
\begin{theorem}\label{prop:log-star}
\Proplogstar
\end{theorem}

If one wants to optimize the communication cost of finding a schedule rather than its length, we show that a (rather lengthy) schedule can be obtained within $O(1)$ rounds (Section~\ref{sec:const-round}).

\begin{theorem}\label{th:conscomm}\textbf{(informal)}
For any graph $G=(V,E)$ and any input of maximal independent sets $\alpha,\beta$ to the MIS-reconfiguration problem, there exists a reconfiguration schedule of length $\Theta(f(n))$, where $f(n)$ is the largest identifier among the nodes in the graph, which can be found in $O(1)$ rounds.
\end{theorem}

The construction generalizes itself on graphs with a distance-$k$ coloring of $c$ colors, with $k$ big enough. It is possible, from this coloring, to compute a schedule of length $O(c)$ after a constant number of communication rounds. Let $\Delta$ be the maximal degree of the graph. A distance-$k$ $O(\Delta^{2k})$ coloring can be found in $O(\log^*n)$ rounds~\cite{Linial:1987:DGA:1382440.1382990}, and a distance-$k$ $O(\Delta^{k})$ coloring can be found in $O(\log^*n+\sqrt{\Delta^k})$ rounds~\cite{BarenboimEG18}. Hence, with the same respective communication complexities, we can find schedules of lengths $O(\Delta^{2k})$ and $O(\Delta^k)$.

Finally, as can be inferred from Theorem~\ref{th:main}, $4$-domination suffices for any graph with diameter greater than $3$. For graphs with small diameter, we give an exact characterization of the conditions that allow the existence of a reconfiguration schedule (Section~\ref{sec:characterization}). This result implies that our algorithm from Theorem~\ref{th:main}, combined with a trivial algorithm that collects the entire graph when the diameter is a small constant, produces an efficient reconfiguration schedule in all cases for which it exists.

\subsection{Related work}

\subparagraph{Distributed Reconfiguration.}
Questions of distributed reconfiguration were actually not studied before 2018. Then, Bonamy et al.~\cite{bonamy2018distributed} considered distributed reconfiguration of colorings, with the goal of finding which length of schedule can be computed within a given number of communication rounds. The problem being PSPACE complete in the general case, several subcases were explored.  Since finding looser restrictions for the transitions is important for making the problem local instead of having to solve a global PSPACE hard problem, the addition of extra colors in the intermediate colorings was allowed. This aided either having a solution, or finding one quickly.

\subparagraph{Distributed Constructions.}
Our constructions sometimes make use of two fundamental subroutines, which find an MIS or a $(3,2)$-ruling set in a graph. An $(x,y)$-ruling set is a set $S \subseteq V$ in which every two nodes are at distance at least $x$, and every node that is not in $S$ is within distance at most $y$ from $S$. Thus, an MIS is a $(2,1)$-ruling set. Finding an MIS is one of the most fundamental problems in distributed computing. The celebrated randomized $O(\log n)$-round algorithms of Luby~\cite{Luby86} and Alon et al.~\cite{AlonBI86} have been recently improved by Ghaffari to $O(\log\Delta+2^{O(\sqrt{\log\log n})})$ rounds, where $\Delta$ is the maximal degree in the network~\cite{Ghaffari16}. Deterministic solutions are the classic network-decomposition based algorithm of Panconesi and Srinivasan that runs in $2^{O(\sqrt{\log n})}$ rounds~\cite{PanconesiS96}, and the $O(\Delta+\log^*n)$-round algorithm of Barenboim et al.~\cite{BarenboimEK14}. The classic lower bound of Linial~\cite{Linial92} shows that $\Omega(\log^*n)$ rounds are necessary, Kuhn et al. gives a higher bound of $\Omega(\log\Delta/\log\log\Delta, \sqrt{\log n/\log\log n})$ ~\cite{KuhnMW16}. The latest results of Balliu et al.~\cite{balliu2019lower} give the new best known lower bounds to find a MIS: There is no deterministic algorithm in $o(\Delta+\frac{\log n}{\log\log n})$ nor randomized algorithm in $o(\Delta+\frac{\log\log n}{\log\log\log n})$. The Figure 1 in~\cite{balliu2019lower} sums up all the results on MIS. A $(3,2)$-ruling set can be computed by computing an MIS over $G^2$, and more general ruling sets have been studied in~\cite{AwerbuchGLP89, SchneiderEW13, PaiPPR017, KuhnMW18, KothapalliP12}.


\subparagraph{Centralized Reconfiguration of Maximal Independent Sets.}
Reconfigurations problems on graphs have been widely studied in the centralized setting during the last decade. An excellent survey on reconfiguration problems can be found in~\cite{Jansurvey}.
In the centralized setting, the transition rules are different, requiring that any intermediate set must be at least of a certain size. While having their own motivation in that setting, these rules are not the ones that are needed in the distributed setting, as they do not give covering guarantees (moreover, such properties would be costly to obtain in a distributed setting, due to their global nature).

In more detail, three kinds of transitions have been studied for the independent set reconfiguration problem.
\emph{Token Addition and Removal}~\cite{ito2010complexity}, or $TAR(k)$: at each transition, one vertex is removed from or added to the current independent set, as long as there are at least $k$ nodes in the independent set.
\emph{Token Jumping}~\cite{kaminski2012complexity}: at each transition, one vertex is removed from the independent set and another one is added somewhere else.
\emph{Token Sliding}~\cite{hearn2002pspace}: at each transition, an edge containing a vertex of the independent set is chosen. This vertex is removed from the set and its neighbor on the other side of that edge is added to the set.
The two first versions are actually equivalent when $k$ corresponds to the size of the independent sets minus 1. Reconfiguration problems are in PSPACE, and independent set reconfiguration problems are in general PSPACE complete~\cite{hearn2002pspace}. Studies over subclasses of graphs exist, and some polynomial algorithm or hardness proofs are given. For example, planar graphs \cite{hearn2002pspace}, perfect graphs \cite{kaminski2012complexity}, trees \cite{demaine2015linear} and bipartite graphs \cite{lokshtanov2018complexity}. 

\section{Preliminaries}
\label{sec:prelim}
We work in the classic LOCAL model of computation, in which $n$ nodes in a synchronous network exchange messages with their neighbors in each round of computation.

Let $G=(V,E,U)$ denote a graph with an assigned subset $U \subseteq V$. An input to the MIS-reconfiguration problem is a pair $G_{input}=(V,E,\alpha),G_{output}=(V,E,\beta)$, where $\alpha$ and $\beta$ are the initial and final maximal independent sets, respectively. We refer to a node $ v \in \alpha$ as an \emph{$\alpha$-node}, and to a node $v \in \beta$ as a \emph{$\beta$-node}.
Notice that a node may be both an  $\alpha$-node and a $\beta$-node.
We refer to node $v \in V\setminus (\alpha\cup\beta)$ as an \emph{$\epsilon$-node}.
Throughout the proofs, we say that a node $v$ is \emph{covered} or \emph{4-dominated} by a node $u$ if $d(v,u) \leq 4$.

For a vertex $v\in V$, we denote by $N(v)$ the set of neighbors of $v$  (i.e., $N(v) = \{u\in V:(u,v)\in E\}$), and given a set $U \subseteq V$ we define $N_U(v)=U\cap N(v)$ for the subset of neighbors of $v$ that are in $U$, and we call this set the $U$-neighbors of $v$. For a subset $U\subseteq Y \subseteq V$ and a node $v \in Y$, we denote by $d_Y(v,U)$ the distance of $v$ from $U$ in the subgraph induced by $Y$.

\begin{definition}[\textbf{Reconfiguration Schedules}]
\label{def:schedule}
For a given property $P$ of $G=(V,E,U)$, an \emph{$(\alpha, \beta, P)$-reconfiguration schedule} (or simply a schedule) $S$ of length $\ell$ is a sequence of subsets of $V$, $S=(S_0,\dots,S_{\ell})$, such that the following hold:
\begin{enumerate}
\item $S_0=\alpha$ and $S_{\ell} = \beta$,
\item for every $0 < i < \ell$, the graph $(V,E,S_i)$ satisfies $P$, and
\item for every $0 < i \leq \ell$, $S_i \oplus S_{i-1}$ is an independent set of $(V,E)$.
\end{enumerate}
\end{definition}


\section{An MIS reconfiguration schedule of constant length}
\label{sec:const-len}
Our main theorem is the following.
\begin{theorem-repeat}{th:main}\textbf{(formal)}
Let $P$ be the property of $(V,E,U)$ that says that $U$ is a (2,4)-ruling set.
For any graph $G=(V,E)$ of diameter greater than 3 and any input $G_{input}=(V,E,\alpha),G_{output}=(V,E,\beta)$ to the MIS-reconfiguration problem, there exists an $(\alpha, \beta, P)$-reconfiguration schedule of constant length $28$.
Moreover, such a schedule can be found in $O(\texttt{MIS}+\texttt{R32})$ rounds, where $\texttt{MIS}$ is the complexity of finding an MIS on a worst-case graph and $\texttt{R32}$ is the complexity of finding a $(3,2)$-ruling set on a worst-case graph.
\end{theorem-repeat}

In particular, Theorem~\ref{th:main} immediately implies a highly efficient solution for bounded degree graphs.
\begin{corollary}
\label{cor:log-star}
The constant length schedule of Theorem~\ref{th:main} can be found in $O(\log^*n)$ rounds in graphs of bounded degree.
\end{corollary}


We now describe the outline of the algorithm, as follows. Denote by $W$ the set of connected components of $\alpha\cup\beta$. Our main approach is to reconfigure the independent sets according to the components in $W$. To this end, we first categorize each component in $W$ according to its diameter and whether it is isolated or not: We say that a component $V_i \in W$ is \emph{isolated} if for every $\epsilon$-node $u$ in its neighborhood, $N_{\alpha}(u)$ and $N_\beta(u)$ are contained in $V_i$.

Notice that within a constant number of rounds, all $\alpha$ and $\beta$-nodes can know whether they are in a component of diameter 0, 1, 2, or at least 3. Moreover, if their diameter is smaller than 3, they can know whether the component is isolated or not.

To avoid excessive notation, we will sometimes say that \emph{we update the component $V_i$ in steps $\{j, j+1\}$}. This means that we remove $\alpha\cap V_i$ from the independent set in step $j$ and we add $\beta\cap V_i$ to the set in step $j+1$. Formally, this means that $S_j = S_{j-1} \setminus (\alpha\cap V_i)$ and $S_{j+1} = S_{j} \cup (\beta\cap V_i)$. Since we will sometimes update multiple components concurrently, we will have $S_j = S_{j-1} \setminus (\alpha\cap Z_j)$ and $S_{j+1} = S_{j} \cup (\beta\cap Z_j)$, where $Z_j = \bigcup_{i \in I_j} V_i$, with $I_j = \{i : V_i \mbox{ is being updated in steps } \{j, j+1\}\}$.

The high-level description of our algorithm is as follows. First, for components in $W$ of diameter 0, we do not need to do anything, as such components are comprised only of nodes in $\alpha\cap\beta$. These nodes remain in the independent set $S_i$ for the entire schedule, and we omit these components and all of their $\epsilon$-neighbors from the remaining discussion. Our algorithm then handles non-isolated components and components of diameter~$\ge3$, and finally handles the isolated components of diameter~$\le2$.

We begin by claiming that with an overhead of 2 rounds, we may assume that $\alpha$ and $\beta$ are disjoint. 
%

\newcommand{\LemmaAlphaBeta}
{
Let $P$ be the property of $(V,E,U)$ that says that $U$ is a (2,4)-ruling set, and let $G=(V,E,\alpha)$ and $G=(V,E,\beta)$ be an input to the MIS-reconfiguration problem. Let $V_{\alpha,\beta}=\alpha \cap \beta$ and denote $V' = V\setminus (V_{\alpha,\beta}\cup N(V_{\alpha,\beta}))$, $\alpha' = \alpha \setminus V_{\alpha,\beta}$ and $\beta' = \beta \setminus V_{\alpha,\beta}$. Let $G' = G[V']$. If there exists an $(\alpha',\beta',P)$-reconfiguration schedule $S'$ of length $\ell$ for $G'$, then there exists an $(\alpha,\beta,P)$-reconfiguration schedule $S$ of length $\ell$ for $G$.

Moreover, any distributed algorithm for finding $S'$ in $T'$ rounds implies a distributed algorithm finding $S$ in $T=T'+2$ rounds.
}
\begin{lemma}
\label{lemma:alpha-beta-nodes}
\LemmaAlphaBeta
\end{lemma}
\begin{proof}
Let $S' = S'_0,\dots,S'_{\ell}$ be an $(\alpha',\beta',P)$-reconfiguration schedule for $G'$.
We define $S_i = S'_i \cup V_{\alpha,\beta}$ for all  $0 \leq i \leq \ell$, and show that $S = S_0,\dots,S_{\ell}$ is an $(\alpha,\beta,P)$-reconfiguration schedule for $G$. Condition (1) of Definition~\ref{def:schedule} holds since $S_0 = S'_0 \cup V_{\alpha,\beta} = \alpha' \cup V_{\alpha,\beta} = \alpha$ and $S_{\ell} = S'_{\ell} \cup V_{\alpha,\beta} = \beta' \cup V_{\alpha,\beta} = \beta$.

For condition (2), fix an index $0\leq i\leq\ell$. Notice that $S'_i$ is a $(2,4)$-ruling set and hence it is independent. Since $S_i \setminus S'_i = V_{\alpha,\beta}$, and $V_{\alpha,\beta}$ is independent, we have that $S_i$ is also independent because no node that is a neighbor of $V_{\alpha,\beta}$ is in $V'$ and in particular no such node is in $S'_i$. In addition, every node in $V'$ is $4$-dominated in $S'_i$, and every node in $V \setminus V'$ is dominated by $V_{\alpha,\beta}$, which implies that $S_i$ is $4$-dominating and hence it is a $(2,4)$-ruling set.

Finally, for each $0\leq i\leq\ell$, we have that $S_i \oplus S_{i-1} = S'_i \oplus S'_{i-1} $ and hence it is an independent set of $(V,E)$.

To obtain $S$ from $S'$ with an overhead of 2 rounds, a distributed algorithm can have each node in $V_{\alpha,\beta}$ indicate this to its neighbors, and then each node in $N(V_{\alpha,\beta})$ can indicate this to its neighbors. Then, the algorithm obtains $S'$ over $G'$ and deduces $S$.
\end{proof}

\subsection{Components of diameter $\geq 3$}

We continue with the following lemma, which is useful for handling components in $W$ whose diameter is not too small. Roughly speaking, the way we handle components of sufficient diameter is by finding a set of $\alpha$-nodes that are not too close to each other to ensure that $\beta$-nodes can be added not too far from them before we remove them from the independent set. This way, we can reconfigure the rest of the component, and then this set of $\alpha$-nodes and their neighbors. We present the following lemma before the rest of the algorithm because we will need to use it, but notice that it is not the case that we begin the algorithm by reconfiguring components of diameter $\geq 3$.

\begin{lemma}
\label{lemma:no-epsilon-nodes}
Let $P$ be the property of $(V,E,U)$ that says that $U$ is a (2,4)-ruling set, and let $G=(V,E,\alpha)$ and $G=(V,E,\beta)$ be an input to the MIS-reconfiguration problem such that $\alpha\cap\beta=\emptyset$, the set $Y=\alpha\cup\beta$ is a single connected component of diameter at least $3$, and each $\epsilon$-node is connected to an $\alpha$-node and to a $\beta$-node.
Then, there exists an $(\alpha,\beta,P)$-reconfiguration schedule of length $8$. Moreover, such a schedule can be found in $O(\texttt{R32})$ rounds.
\end{lemma}
\begin{proof}
First, assume that the diameter of $Y=\alpha\cup\beta$ is either $3$ or $4$.
Consider a shortest path of length $3$ in $Y$, denoted by $(v_1,v_2,v_3,v_4)$. Either $v_1$ or $v_4$ is in $\alpha$, and is within distance $4$ from all other nodes in the component. We denote this node by $v$, and define $S_1 = \{v\}$ and we have that it 4-dominates the entire component. In addition, it 4-dominates all $\epsilon$-nodes, by the assumption of the lemma that all such nodes are neighbors of $\beta$-nodes, because $v$ actually 3-dominates all $\beta$-nodes in the component. We then denote $R = \{u \in \beta~|~ u \not\in N(v)\}$ and define $S_2 = \{v\} \cup R$, and $S_3 = R$, and finally $S_4 = R \cup N(v)$. It is easy to verify that this results in a valid  $(\alpha,\beta,P)$-reconfiguration schedule. In particular, notice that, without loss of generality, if $v=v_1$, then $R$ contains at least the node $v_4$, which 4-dominates the entire component.

For a diameter of $Y$ that is at least 5, the high-level idea of the construction is as follows. Consider a $(3,2)$-ruling set $R$ over the nodes in $\alpha$, where we imagine an edge between two nodes in $\alpha$ if they are at distance two in the subgraph induced by $Y$. We reconfigure all $\beta$-nodes that are at distance $5$ from $R$ in $G$ by removing their $\alpha$-neighbors first, then by adding them. Then, we do the same for $\beta$-nodes that are at distance $3$ from $R$, and finally we repeat this one last time for the $\beta$-nodes in the direct neighborhood of $R$. The choice of a $(3,2)$-ruling set ensures that all $\alpha$-nodes in $R$ have a $\beta$-node at distance 3 that will be reconfigured in the 4th step.  However, while we trust $\beta$-nodes at distance 3 from $R$ to cover $\alpha$-nodes at distance 2 from $R$ while $R$ itself is being reconfigured, we must be careful when handling $\alpha$-nodes at distance 2 from $R$ that do not have a neighbor at distance 3. We overcome this caveat by taking care of these nodes separately.

Formally, we define a virtual multigraph $\tilde{G}=(\tilde{V},\tilde{E})$ as follows. The set of virtual nodes $\tilde{V}$ consists of all $\alpha$-nodes. If $v$ and $u$ in $\tilde{V}$ have a common $\beta$-neighbor, we add an edge ${u,v}$ to $\tilde{E}$. Let $R$ be a $(3,2)$-ruling set in $\tilde{G}$. It is easy to see that in $G$, the set of nodes $R$ is a $(6,5)$-ruling set of $Y$. We denote $R$ by $R_0$ and we define $R_i$ for $3 \leq i \leq 5$ as $R_i = \{v\in Y ~|~ \mbox{the distance of $v$ from $R$ in the subgraph induced by $Y$ is $i$}\}$. Then we define
$R_{1} = \{v\in Y ~|~ d_Y(v,R)=1 \mbox{ and } d_Y(v,R_3)=2\}\cup\{v\in Y ~|~ N_Y(v)\subseteq R\}$,
which captures all $\beta$-neighbors of $R$ that either do not have other $\alpha$-neighbors, or have other $\alpha$-neighbors which in turn have $\beta$-neighbors that are farther from $R$. We separate those from the set $R_{-1}= N_Y(R_0)\setminus R_1$. We complete the partition by defining $R_2=N_Y(R_1)\setminus R_0$ and $R_{-2}= N_Y(R_{-1})\setminus R_0$. Note that for even $i$, $R_i$ contains only $\alpha$-nodes, and for odd $i$, $R_i$ contains only $\beta$-nodes. We have that, for every $-2\le i\le5$, $N_Y(R_i)\subseteq R_{i-1}\cup R_{i+1}$ (with $R_{-3}=R_6=\emptyset)$. By construction of $R$, we have that each node in $R$ has a node at distance 3 in $R_3$, hence it has a node at distance $2$ in $R_2$ and a node at distance 1 in $R_1$.

%

We define $S_0=\alpha$ and for $i=0,1,2,3$, we define $S_{2i+1} = S_{2i} \setminus R_{4-2i}$, $S_{2(i+1)} = S_{2i+1} \cup R_{5-2i}$.
We claim that $S_0,\dots,S_8$ is an $(\alpha,\beta,P)$-reconfiguration schedule.

First, $S_0=\alpha$ by definition, and because every $\beta$-node is within an odd distance of at most $5$ from $R$ and every $\alpha$-node is within an even distance of at most $4$ from $R$, we have that $S_8=\beta$. This gives condition (1) of Definition~\ref{def:schedule}.

For condition (2), it is easy to see that $S_i \oplus S_{i-1}$ is an independent set of $(V,E)$ for every $1\leq i \leq 8$. For an odd $i$ this holds because to obtain $S_i$ we only remove $\alpha$-nodes from $S_{i-1}$, and no two such nodes can be neighbors. For an even $i$ this holds because to obtain $S_i$ we only add $\beta$-nodes to $S_{i-1}$, and no two such nodes can be neighbors.

It remains to show condition (3) of Definition~\ref{def:schedule}.
To show that $S_{i}$ is independent for $i=2,4,6$, notice that $\beta$-nodes in $R_j$ (for $j=-1,1,3,5$) are only added to the sequence after all $\alpha$-nodes in $R_{k}$ for $k\geq j-1$ have been removed. By definition, $S_0$ is also independent. Hence, for $i=1,3,5,7$, $S_i$ is independent because it is a subset of $S_{i-1}$.

Next, we need to show that $S_{i}$ is $4$-dominating for every $1 \leq i \leq 7$.  Our focus will be for $i=1,3,5,7$, and for $i=2,4,6$ it then follows because $S_i$ contains $S_{i-1}$. We first show it on $Y$, and will prove it for $\epsilon$-nodes afterward. For $i=1$ this holds because all nodes in $R_j$ for $j\leq 3$ are in or have neighbors in $R_{-2}\cup R_0\cup R_2$. All nodes in $R_j$ for $j=4,5$ are within distance $3$ from $R_2$.
Similarly, $S_{3}$ is $4$-dominating because nodes in $R_{j}$ for $j\leq4$ are covered by $R_0$, nodes in $R_{5}$ are in the current independent set.
For $S_{5}$, recall that for any node in $R$, there is a node at distance 3 from it in $R_3$, that node being in the current independent set since $S_4$. Hence, $R_3$ covers $R_{j}$ for $-1\le j\le 5$. $R_{-2}$ is still included in $S_{5}$.
Finally, for $S_{7}$, $R_3$ still covers $R_{j}$ for $-1\le j\le 5$. For each node in $R_{-2}$, there is a node at distance 3 from it in $R_1$ that has been added in $S_6$ that covers it.

Now, let $u$ be an $\epsilon$-node that has a node $a\in\alpha$ and $b\in\beta$ in its neighborhood. We show that $a$ or $b$ are always 3-dominated throughout the sequence. In a step where $b$ has no $\alpha$-neighbor in the independent set, it must be a step right before $b$ gets added to the independent set. If $b$ is in $R_5$ or in $R_3$ then when this happens, it is 3-dominated by an $\alpha$-node in $R_2$ or $R$, respectively, and this node is still in the independent set. If $b\in R_{-1}$ then it is 2-dominated by nodes in $R_{-2}$ and then $R_1$ (with an overlap in $S_6$, the construction ensures that such node exist at distance at most 3 from $b$). Finally, If $b \in R_1$ then either there is a $\beta$-node in $R_3$ that 2-dominates it, and this node is already in the independent set, or $b$ is in $\{v\in Y ~|~ N_Y(v)\subseteq R\}$. Only in the latter case, we must resort to the $\alpha$-neighbor of $u$ and check that it is $3$-dominated by $S_5$, as we removed $R$ from $S_5$ and $b$ is added in $S_6$.

Let $i$ be such that $a\in R_i$. We need to make sure that $a$ is 3-dominated at the step in which we reconfigure $R_1$. At this step, all of the $\beta$-nodes in $R_3, R_5$ are in the independent set, and hence their $\alpha$-neighbors in $R_2,R_4$ are covered by nodes in distance 1, and nodes in $R_0$ are covered by nodes in distance 3. For $\alpha$-nodes in $R_{-2}$ they are still in the independent set at this step, and hence are 3-dominated.

This completes the correctness proof.
For the round complexity, notice that simulating the $(3,2)$-ruling set over $\tilde{G}$ can be done in $G$ with a constant overhead.
\end{proof}

\subsection{Non-isolated components}
We first observe that components of diameter $\le2$ are such that there is a complete bipartite graph between their $\alpha$-nodes and $\beta$-nodes.
Let $u$ be an $\epsilon$-node that is a neighbor of several components. Let $W_u$ be the set of all components that are its neighbors, so that in particular, $V_i, V_j \in W_u$. For each pair of distinct components $V_i, V_j \in W_u$, if there is an $\alpha$ node in $N_\alpha(u)\cap V_i$ and  a $\beta$ node in $N_\beta(u)\cap V_j$, then we say that $V_j$ is $(u,\alpha)$-covered and that $V_i$ is $(u,\beta)$-covered (note that this definition allows a single component to satisfy both conditions). As $u$ is an $\epsilon$-node, there must exist a component $V_{u,\alpha} \in W_u$ that is $(u,\beta)$-covered and a component $V_{u,\beta} \in W_u$ that is $(u,\alpha)$-covered.

We say that a component $V_i \in W$ is $\alpha$-covered ($\beta$-covered) if there is an $\epsilon$-node $u$ for which $V_i$ is $(u,\alpha)$-covered ($(u,\beta)$-covered). A component that is both is 
$\alpha\beta$-covered.

The key insight is that a $(u,\alpha)$-covered component of diameter~$\le2$ is covered (dominated at distance 4) by some $\alpha$-node of the component $V_{u,\beta}$ (and similarly with the $\beta$-node of $V_{u,\alpha}$). Moreover, any $\epsilon$-node that is connected to an $\alpha$-node (a $\beta$-node) in that component is covered by $V_{u,\beta}$ (or $V_{u,\alpha}$). This implies that an $\epsilon$-node that is connected to two components that are updated in different steps is always covered by the component that is currently not being updated.
However, during the reconfiguration schedule, we need to be careful about $\epsilon$-nodes that are connected to a single component, and $\epsilon$-nodes that are connected to two components that are updated at the same time.

We denote by $C_{\alpha\beta}$ the set of $\alpha\beta$-covered components of diameter~$\le2$, and by $C_\alpha$ and $C_\beta$ the sets of $\alpha$-covered and $\beta$-covered components of diameter~$\le2$ that are not in $C_{\alpha\beta}$, respectively.
Define the component graph $\tilde{G} = (W,\tilde{E})$, where there is an edge between $V_i,V_j \in W$ iff there exists an $\epsilon$-node $u$ such that $V_i$ is $(u,\alpha)$-covered and $V_j$ is $(u,\beta)$-covered, or vice-versa.
Notice that in $\tilde{G}$, the sets $C_\alpha$ and $C_\beta$ are two disjoint independent sets.

We are finally ready to formally provide the algorithm for handling all components that are either non-isolated or have diameter $\geq 3$.

\begin{lemma}
\label{lemma:non-isolated}
Let $P$ be the property of $(V,E,U)$ that says that $U$ is a (2,4)-ruling set, and let $G=(V,E,\alpha)$ and $G=(V,E,\beta)$ be an input to the MIS-reconfiguration problem such that $\alpha\cap\beta=\emptyset$, and all connected components of $\alpha\cup\beta$ are either non-isolated or have diameter at least 3. Then, there exists an $(\alpha,\beta,P)$-reconfiguration schedule of length $18$. Moreover, such a schedule can be found in $O(\texttt{MIS}+\texttt{R32})$ rounds.
\end{lemma}

\begin{proof}
Our reconfiguration schedule works according to the following \emph{parts}.
\begin{enumerate}
\item Update components of diameter~$\le2$ in $C_\alpha$.
~\\\%Let $M$ be an MIS over all nodes in $C_{\alpha\beta}$.
\item Update components of diameter~$\le2$ that are $\alpha$-covered by a component in $M$.
\item Reconfigure components of diameter~$\ge3$ using the schedule given by Lemma~\ref{lemma:no-epsilon-nodes}.
\item Update components in $M$.
\item Update components of diameter~$\le2$ in $C_{\alpha\beta}$ that were not previously updated.
\item Update components of diameter~$\le2$ in $C_\beta$ that were not previously updated.
\end{enumerate}

First, it is easy to see that the schedule has length 18. The part that reconfigures components of diameter~$\ge3$ requires 8 steps, by Lemma~\ref{lemma:no-epsilon-nodes}. Each of the other 5 parts takes exactly 2 steps as described in the definition of updating components (removing $\alpha$-nodes and then adding $\beta$-nodes), which sums to 18 reconfiguration steps in the schedule.

It remains to prove correctness. First, condition (1) of Definition~\ref{def:schedule} trivially holds, as the schedule reconfigures all nodes. Moreover, by Lemma~\ref{lemma:no-epsilon-nodes} and by the definition of updating a component, it is also immediate that we do not reconfigure two neighbors in a single step, thus the schedule satisfies condition (3) of Definition~\ref{def:schedule}. For condition (2), Lemma~\ref{lemma:no-epsilon-nodes} and the definition of updating a component also guarantee that each $S_i$ is an independent set. The remainder of the proof shows that each $S_i$ in the schedule is also $4$-dominating.

By the order of the reconfiguration steps in the schedule, each component that is being updated is covered by a component that is not concurrently being updated. This also holds for $\epsilon$-nodes that are connected to a component that is not currently being updated. The main condition that must be verified is that $\epsilon$-nodes remain covered even if all of their neighboring components are being concurrently updated in a certain part of the schedule.

Part 1 guarantees that $S_1,S_2$ are 4-dominating because for each component that is being updated, the $\alpha$-node covering it is a member of $S_1,S_2$ and it also covers the required $\epsilon$-nodes that are neighbors of the updated component, as explained earlier.
For part 2, let $u$ be an $\epsilon$-node that is connected to two of the components that are being updated and is not connected to any component that is not being updated. One of the components must be connected to $u$ via an $\alpha$-node. Let $V_i$ be such a component, let $u_1$ be the $\alpha$-node connected to $u$,
 and let $v$ be the $\alpha$-node from a component of $M$ that covers $V_i$. The distance between $v$ and $u_1$ is 3: $v$ is a distance 2 to a $\beta$-node of $V_i$ and, because $V_i$ is of diameter~$\le2$, within $V_i$ all $\beta$-nodes are connected to all $\alpha$-nodes. Hence, $u$ is at distance 4 from $v$.

For part 3, the 4-domination within the components that are being reconfigured is given by Lemma~\ref{lemma:no-epsilon-nodes}. Notice that any $\epsilon$-node connected to a component of diameter~$\ge3$ is connected either to connected components of diameter~$\ge3$ through both an $\alpha$ and a $\beta$-node, or to a component that is not being updated in those steps. In the first case it is covered by Lemma~\ref{lemma:no-epsilon-nodes}, and in the second it is covered by the component not being concurrently updated.

For part 4, notice that all components that are $\beta$-covering components of $M$ have been updated in steps 2 or 3. Hence, as each component of $M$ is in $C_{\alpha\beta}$, there is a $\beta$-node in the current set $S_{12}$ that covers it. As $M$ is independent, we do not have $\epsilon$-node between two components that are being updated.
For part 5, the $\epsilon$-nodes between two components that are being updated are covered by an argument that is symmetric to the one used for part 2.
Finally, for part 6, for each component that is being updated it holds that the $\beta$-node covering it is in a component that has already been updated and hence it is already in $S_{16}$.

Finally, we note that apart from a constant overhead in communication, the number of rounds required for computing the above schedule is proportional to that of finding the MIS $M$ plus solving the diameter $\ge3$ components, which completes in $O(\texttt{MIS}+\texttt{R32})$ rounds, where $\texttt{MIS}$ is the complexity of finding an MIS on a worst-case graph and  and $\texttt{R32}$ is the complexity of finding a $(3,2)$-ruling set on a worst-case graph, as claimed.
\end{proof}

\subsection{Isolated Components}
What remains now is to handle components that are isolated and have diameter at most 2. When we address these components, we will also address all of their $\epsilon$-neighbors. Hence, from this point onwards we will slightly abuse our terminology, and when we refer to such a component we refer to its nodes along with their $\epsilon$-neighbors as the component. This means that now the components that we address might have a diameter that is increased by 1, and thus their diameter can be also 3. Note that the diameter cannot be increased by two as all $\alpha$-nodes are connected to all $\beta$-nodes, and each $\epsilon$-node is connected to an $\alpha$-node and to a $\beta$-node of this component, otherwise the component would not be isolated.

By definition of isolated components, the neighborhood of an $\epsilon$-nodes within such a component, besides containing vertices of the component itself, is only composed of other $\epsilon$-nodes. Moreover, there is at least one additional $\epsilon$-node in this neighborhood, as we consider graphs of diameter at least 4. We distinguish two kinds of isolated components, according to whether their diameter is at most 2, or whether it is 3.

For a component $V_i$ of diameter~$\le2$, suppose $u$ is an $\epsilon$-node that is a neighbor of $V_i$. This node $u$ has an $\alpha$-node and a $\beta$-node in its neighborhood, that both cover the entire component. Therefore, to update such components, it suffices to make sure that a non-$\epsilon$ neighbor of $u$ is in the current independent set during the two reconfiguration steps. By considering connected two of those components that cover each other, we can take an MIS $M$ over those. The schedule of length 4 is: update $M$, and then update the other components.

Assume now that $V_i$ is a component of diameter 3. It holds that there exists an $\epsilon$-node $u$, an $\alpha$-node $a$ and a $\beta$-node $b$ such that $(u,a)\not\in E$ and $(u,b)\not\in E$ (otherwise the diameter would be 2). Here is an informal description of a schedule of 6-steps for this component.
\begin{enumerate}
\item Remove $N_\alpha(u)$. The node $a$ stays in the independent set and covers the entire component.
\item Add $u$ in the set.
\item Remove the remaining $\alpha$-nodes of the component. The node $u$ covers everything.
\item Add $b$ to the set. Note that $b$ covers the component.
\item Remove $u$.
\item Add the remaining $\beta$-nodes of the component.
\end{enumerate}

A caveat is encountered in case there are two such components, $V_1$ and $V_2$, whose selected $\epsilon$-nodes, $u_1$ and $u_2$, are connected. In such case we cannot do the above 6-step schedule in parallel without violating independence. However, observe that if a single of those two $\epsilon$-nodes is added to the set, it actually covers the second component as well, as it has a diameter of 3. As a consequence, taking an MIS over those $\epsilon$-nodes gives us a selection of nodes that cover all the considered components. Hence, consider the schedule above as being for component $V_1$ and denote $u=u_1$, then we can add the following to steps 3 and 4 above:
\begin{itemize}
\item[3.] Remove the remaining $\alpha$-nodes of $V_1$ and all $\alpha$-nodes of $V_2$. The node $u$ covers everything.
\item[4.] Add $b$ and the $\beta$-nodes of $V_2$ to the set. Note that $V_2$ is updated and $b$ covers $V_1$.
\end{itemize}

We now formalize the above intuition in order to prove the following.
\newcommand{\LemmaIsolated}
{
Let $P$ be the property of $(V,E,U)$ that says that $U$ is a (2,4)-ruling set, and let $G=(V,E,\alpha)$ and $G=(V,E,\beta)$ be an input to the MIS-reconfiguration problem such that $\alpha\cap\beta=\emptyset$, and all connected components of $\alpha\cup\beta$ are isolated and have diameter at most 2. Then, there exists an $(\alpha,\beta,P)$-reconfiguration schedule of length $10$. Moreover, such a schedule can be found in $O(\texttt{MIS})$ rounds.
}
\begin{lemma}
\label{lemma:isolated}
\LemmaIsolated
\end{lemma}

\begin{proof}
We say that two isolated components of diameter~$\le2$ are connected if two of their $\epsilon$-nodes are connected.
Our reconfiguration schedule works according to the following \emph{parts}.
\begin{enumerate}
\item \%Let $M$ be an MIS over isolated components of diameter~$\le2$.
\begin{enumerate}
\item Update the components in $M$.
\item Update the other components of diameter~$\le2$. 
\end{enumerate}
\item \%For each component $V_i$ of diameter 3, we select an $\epsilon$-node $u_i$, an $\alpha$-nodes $a_i$ and a $\beta$-node $b_i$ as described above. That is, $(u_i,a_i)\not\in E$ and $(u_i,b_i)\not\in E$. Let $M'$ be an MIS over the nodes $\{u_i\}_{i\in I}$. Note that if $u_i$ is not in $M'$, it has a neighbor $u_j$ that is in $M'$, and $u_j$ covers $V_i$. Let $S_k$ be the current set of the schedule (we will have that $k=4$).
\begin{enumerate}
\item $S_{k+1} = S_k \setminus \bigcup_{u_i \in M'} N_\alpha(u_i)$. \% Remove all $\alpha$-neighbors of the MIS nodes.
\item $S_{k+2} = S_{k+1} \cup M'$. \%Add all the MIS nodes. 
\item $S_{k+3} = S_{k+2} \setminus (\alpha \cap(\bigcup V_i))$. \%Remove all the $\alpha$-nodes of all the $V_i$.
\item $S_{k+4} = S_{k+3} \cup (\bigcup_{i: u_i \in M'} b_i \bigcup_{j: u_j \not\in M'} (\beta\cap V_j))$. \\\%For each $u_i\in M'$, add $b_i$. For each $u_j\not\in M'$, add $\beta$-nodes of $V_j$.
\item $S_{k+5} = S_{k+4} \setminus M'$. \%Remove $M'$.
\item $S_{k+6} = S_{k+5} \cup (\beta \cap(\bigcup V_i))$.\%Add the remaining $\beta$-nodes.
\end{enumerate}
\end{enumerate}

For the length of the schedule, notice that parts 1a and 1b require two steps each, by the definition of updating a component. Part 2 requires 6 steps, summing to 10 steps in total.

For correctness, condition (1) of Definition~\ref{def:schedule} trivially holds, as the schedule reconfigures all nodes. Moreover, by the definition of updating a component, and by inspection of part 2, it is also immediate that we do not reconfigure two neighbors in a single step, thus the schedule satisfies condition (3) of Definition~\ref{def:schedule}. For condition (2), the definition of updating a component and inspection of part 2 also guarantee that each $S_i$ is an independent set. The remainder of the proof shows that each $S_i$ in the schedule is also $4$-dominating.

For parts 1a and 1b, note that having an MIS over the components promises that whenever a component gets updated it has a neighboring component that stays untouched.

For part 2, the arguments are the following. The sets $S_6, S_8, S_{10}$ are supersets of $S_5, S_7, S_9$, respectively, and therefore we only need to show domination for the latter. In $S_{5}$ we have that all the nodes $a_i$ dominate the components, as they are not connected to the removed nodes $u_i$. In $S_7$ we have that the nodes in $M'$ cover everything. Finally, in $S_9$ we have that the nodes $b_i$ cover all the components.

For the number of communication rounds, it is easy to see that the schedule can be found in $O(\texttt{MIS})$ rounds, where $\texttt{MIS}$ is the complexity of finding an MIS on a worst-case graph, as needed.
\end{proof}

\subsection{Completing the proof}
We can now wrap-up all the ingredients and prove Theorem~\ref{th:main}.
\begin{proof}[\textbf{Proof of Theorem~\ref{th:main}}]
We describe the full $(\alpha,\beta,P)$-reconfiguration schedule $S$.
First, each node $v$ in $V_{\alpha,\beta} = \alpha \cap \beta$ sends a message to its neighbors in $N(v)$ and outputs that it is in $S_i$ for all $0 \leq i \leq 28$. Each node that received such a message, sends a message to its neighbors and outputs that it is not in $S_i$ for all $0 \leq i \leq 28$. By Lemma~\ref{lemma:alpha-beta-nodes}, this is consistent with any reconfiguration schedule for the rest of the nodes. The nodes that produced an output terminate and any edges incident to them are removed from the graph.

Next, all nodes collect their 4-hop neighborhood to decide whether they are in a component of diameter $\geq 3$ or not, and if not then whether they are in an isolated component.

The components of diameter $\geq 3$ and the non-isolated components compute the reconfiguration schedule of 18 steps, as given in Lemma~\ref{lemma:non-isolated}, which we denote by $S'_0,\dots,S'_{18}$. The isolated components of diameter $\leq 3$ compute the reconfiguration schedule of 10 steps, as given in Lemma~\ref{lemma:isolated}, which we denote by $S''_0,\dots,S''_{10}$.

Formally, the $(\alpha,\beta,P)$-reconfiguration schedule is now $S_i = S''_0\cup S'_i \cup V_{\alpha,\beta}$ for $0 \leq i \leq 18$ and $S_i = S''_{i-18} \cup S'_{18} \cup V_{\alpha,\beta}$ for $18 \leq i \leq 28$. It is computed within $O(\texttt{MIS}+\texttt{R32})$ rounds.
\end{proof}

\section{MIS reconfiguration in a constant number of rounds}
\label{sec:const-round}

\begin{theorem-repeat}{th:conscomm}\textbf{(formal)}
Let $P$ be the property of $(V,E,U)$ that says that $U$ is a (2,4)-ruling set.
For any graph $G=(V,E)$ and any input $G_{input}=(V,E,\alpha),G_{output}=(V,E,\beta)$ to the MIS-reconfiguration problem, there exists an $(\alpha, \beta, P)$-reconfiguration schedule of  length $\Theta(f(n))$, where $f(n)$ is the largest identifier among the nodes in the graph, which can be found in $O(1)$ rounds.
\end{theorem-repeat}

To prove this, we first prove the following lemma, stating that we can always reconfigure locally an independent set to add elements from $\beta$ without losing any element in $\alpha\cap\beta$.

\newcommand{\LemmaConstRounds}
{
Let $P$ be the property of $(V,E,U)$ that says that $U$ is a (2,4)-ruling set.
For any graph $G=(V,E)$ of diameter greater than 5, two MIS $\alpha$ and $\beta$ and $v\in\beta\setminus\alpha$, there exists an MIS $\gamma$ such that
\begin{enumerate}
\item $(\alpha\cap\beta)\subset\gamma$ and $v\in\gamma$, and
\item there exists an $(\alpha, \gamma, P)$-reconfiguration schedule of  length 6.
Moreover, for finding the reconfiguration schedule the nodes only need to know the topology of their 5-hop neighborhood and therefore can be found in $O(1)$ rounds.
\end{enumerate}
}

\begin{lemma}
\label{lemma:constant-rounds}
\LemmaConstRounds
\end{lemma}

\begin{proof}[\textbf{Proof of Lemma~\ref{lemma:constant-rounds}}]
Note that because the diameter is greater than 5, there are always nodes at distance 3 from any node. Our schedule will always be an alternation of removing nodes, and adding independent nodes to the set, which ensures property (2) of the definition.
Each time we want to reconfigure a node $v$ to be added to the independent set, we first remove its neighbors in $N_\alpha(u)$, which ensures the independence part of the property (2). To be able to do so, we always ensure that neighbors at distance 2 from $v$ are covered (4-dominated), as those include the neighbors of $N_\alpha(v)$, which permits to satisfy the covering property of (2). For an independent set $X$, we define $GC(X)$ as the greedily completion of $X$ to be an MIS. Note that if the longest distance between two nodes that are added is constant, this can be done in a constant number of steps. Note that we will give priority to $\beta$-nodes and then $\alpha$-nodes during the completion.

First, consider the case where there exists an $\alpha$-node $u$ at distance 2 from $v$. Notice that $u$ covers all the nodes at distance $2$ from $v$. We can remove $N_\alpha(v)$ from the independent set, then add $v$ to it, and finally complete the current independent set greedily to be maximal. Formally, we define $\gamma=GC(\{v\}\cup\alpha\setminus N_\alpha(v))$ and we have $S_0=\alpha$, $S_1=S_0\setminus N_\alpha(v)$, $S_2=S_3=S_4=S_5=S_6=\gamma$.

From now on, we consider the case that only $\beta$-nodes and $\epsilon$-nodes are at distance 2 from $v$.
We differentiate several cases depending on the neighborhood of nodes at distance 2 from $v$:
\begin{enumerate}
\item There exists a node $u$ at distance 2 from $v$ such that $N(u)\cap N_\alpha(v)=\emptyset$:

Let $u_1$ be a node in $N(v)\cap N(u)$. Then $u_1$ is an $\epsilon$-node, as it cannot be an $\alpha$-node by the assumption, nor a $\beta$-node (as $v$ is a $\beta$-node). The node $u_1$ must have a neighbor in $N_\alpha(v)$, as the nodes in $N_\alpha(u_1)$ cannot be in distance 2 from $v$.
That $\alpha$-node covers $u$ and the nodes at distance 2 from $u$. Hence, we can remove all nodes in $N_\alpha(u)$: $S_1=\alpha\setminus N_\alpha(u)$. We can then add $u$ to the independent set: $S_2=S_1\cup\{u\}$. Now, $u$ covers  $N_\alpha(v)$ and its neighborhood, so we can remove those nodes from the independent set: $S_3=S_2\setminus N_\alpha(v)$. Next, $v$ can be added to the independent set and then $u$ and its neighbors are covered by $v$: $S_4=S_3\cup\{v\}$. Hence, we can then remove $u$ from the independent set, and finally add $N_\alpha(u)$ back to the independent set: $S_5=S_4\setminus\{u\}$, $S_6=\gamma=GC(S_5\cup N_\alpha(u))$.

\item There exists a node $u$ at distance 2 from $v$ such that $N_\alpha(v)\subseteq N(u)$ and $u$ has a neighbor $u_3$ that is at distance 3 from $v$:

Notice that $u_3$ covers all nodes in $N_\alpha(v)$ and their neighbors. Symmetrically, all nodes in $N_\alpha(v)$ cover $u_3$ and its neighbors at distance 2.
First we remove all nodes in $N_\alpha(u_3)$ and then add $u_3$ to the independent set (we do nothing if $u_3$ itself is in $\alpha$): $S_1=\alpha\setminus N_\alpha(u_3)$, $S_2=S_1\cup\{u_3\}$. We then remove $N_\alpha(v)$ from the independent set, and then add $v$: $S_3=S_2\setminus N_\alpha(v)$, $S_4=S_3\cup\{u\}$. Finally, we can put $u_3$ and its neighborhood back in their initial states:  $S_5=S_4\setminus\{u\}$, $S_6=\gamma=GC(S_5\cup N_\alpha(u))$.

\item There exists a node $u$ at distance 2 from $v$  that has no neighbor at distance 3 from $v$, and there exists a node $a$ in $N_\alpha(v)$ that is not in the neighborhood of $u$:

We can remove $N_\alpha(u)$, as $a$ covers those nodes and their neighborhood, and we can then add $u$: $S_1=\alpha\setminus N_\alpha(u)$, $S_2=S_1\cup\{u\}$. Now, we can remove what remains of $N_\alpha(v)$, and then add $v$: $S_3=S_2\setminus N_\alpha(v)$, $S_4=S_3\cup\{v\}$. Finally, we remove $u$: $S_5=S_4\setminus\{u\}$, $S_6=\gamma=GC(S_5)$.

\item Now we are in a situation where all nodes at distance 2 from $v$ that have no neighbor at distance 3 from $v$ have $N_\alpha(v)$ in their neighborhood. Let's call $U_2$ the set of nodes at distance 2 from $v$ with a neighbor at distance 3 from $v$. For each $u\in U_2$, we call $u_3$ one of its neighbor which is at distance 3 from $v$. We know that $u$ has at least one node in $N_\alpha(v)$ in its neighborhood, and we will call it  $a_u$. The node $a_u$ covers $u_3$ and its neighborhood at distance 2. We can remove $N_\alpha(u_3)$ and add $u_3$ for all $u\in U_2$: $S_1=\alpha\setminus\bigcup\limits_{u\in U_2}N_\alpha(u_3)$, $S_2=S_1\bigcup\limits_{u\in U_2}\{u_3\}$.

We need now to prove that each vertex in $N_\alpha(v)$ and its neighborhood is covered by some $u_3$. It is true for $N_\alpha(v)$ as each $u_3$ is at distance 3 from $v$. It is also true for $v$. It remains to show this for the nodes at distance 2 from $v$. For those who have a neighbor at distance 3, they are in $U_2$ and have by construction a $u_3$ in their direct neighborhood. For the others, these nodes have $N_\alpha(v)$  in their neighborhood, otherwise we would be in case 3. Hence, they are at distance at most 4 from any $u_3$.

We can remove $N_\alpha(v)$  from the independent set, and we can then add $v$: $S_3=S_2\setminus N_\alpha(v)$, $S_4=S_3\cup\{u\}$. We then remove all the $u_3$ nodes and put back their neighborhood:  $S_5=S_4\setminus\bigcup\limits_{u\in U_2}\{u_3\}$, $S_6=\gamma=GC\left(S_5\bigcup\limits_{u\in U_2} N_\alpha(u_3)\right)$.
\end{enumerate}

Note that in this process, the only nodes actually removed from the independent set are those in $N_\alpha(v)$, are we have put back all the other ones. As $v\in\beta$, it means that $N_\alpha(v)\cap\beta=\emptyset$. Hence, we did not remove any node from $\alpha\cap\beta$.
\end{proof}

Lemma~\ref{lemma:constant-rounds}, 
means that for any element $v$ in $\beta$, we can add $v$ to the current MIS in a constant number of steps without losing any element of $\beta$ already in the MIS. It allows us to prove Theorem~\ref{th:conscomm} as follows.

\begin{proof}[\textbf{Proof of Theorem~\ref{th:conscomm}}]
Nodes use their identifiers to know when to start their own reconfiguration. A node with identifier $k$ uses slots $[6k+1, 6(k+1)]$ for its schedule.
Since a node only needs to know its 5-hop neighborhood, this completes in $O(1)$ rounds.
\end{proof}

If the identifiers are guaranteed to be $\{1,\dots, n\}$ then Theorem~\ref{th:conscomm} gives that a constant number of rounds is sufficient for a linear length schedule. However, we can do even better by using coloring algorithms, as stated in the following corollary.
\begin{corollary}
\label{cor:coloring}
Let $P$ be the property of $(V,E,U)$ that says that $U$ is a (2,4)-ruling set.
For any graph $G=(V,E)$ and any input $G_{input}=(V,E,\alpha),G_{output}=(V,E,\beta)$ to the MIS-reconfiguration problem, if the nodes are given a $k$-coloring of $G^{10}$, then there exists an $(\alpha, \beta, P)$-reconfiguration schedule of  length $O(k)$, which can be found in $O(1)$ rounds.
\end{corollary}

\section{A complete characterization for the existence of a reconfiguration schedule with 4-domination}
\label{sec:characterization}

The following gives an exact characterization of inputs for which there exists a reconfiguration schedule with 4-domination.
In what follows, we say that two sets of nodes $U_1$ and $U_2$ are \emph{fully connected} if every node in $U_1$ is a neighbor of every node in $U_2$. If $U_1$ contains only a single node, then we simply say that this node is fully connected to $U_2$. 

\newcommand{\ThmChar}
{
Let $P$ be the property of $(V,E,U)$ that says that $U$ is a (2,4)-ruling set.
For any input $G_{input}=(V,E,\alpha),G_{output}=(V,E,\beta)$ to the MIS-reconfiguration problem, there does not exists an $(\alpha, \beta, P)$-reconfiguration schedule if and only if:
\begin{enumerate}
\item The sets $\alpha$ and $\beta$ are fully connected.
\item
Let $\epsilon_\alpha$ (resp. $\epsilon_\beta$) be the set of $\epsilon$-nodes that are fully connected to $\alpha$ (resp. $\beta$). Then all the $\epsilon$-nodes are in $\epsilon_\alpha\cup\epsilon_\beta$.
\item Let $G'=(V'=\epsilon_\alpha\cup\epsilon_\beta,E'=\overline{E_{V'}}) $, where $\overline{E_{V'}}$ is the complementary of $E$ restricted to vertices of $V'$. Then there is no path from $\epsilon_\alpha\setminus\epsilon_\beta$ to $\epsilon_\beta\setminus\epsilon_\alpha$ in $G'$.
\end{enumerate}
}
\begin{theorem}
\label{thm:characterization}
\ThmChar
\end{theorem}

\begin{proof}
\textbf{If there is no schedule then Conditions (1), (2), and (3) hold:}
Let  $G=(V,E,\alpha)$ and $G=(V,E,\beta)$ be such that there is no $(\alpha,\beta,P)$-reconfiguration schedule. First, since Theorem~\ref{th:main} produces a schedule when the diameter is greater than 3, we know that the diameter of $G$ is bounded by $3$. This implies the powerful property that it is sufficient for a single node to be part of the set, for any intermediate independent set, because such a node clearly $4$-dominates the rest of the graph.

Since there is no reconfiguration schedule, for any sequence of independent sets that are $4$-dominating, there is a node in $\beta$ that is not added to the set. But we can say something stronger, which is that for \emph{any} node $v\in\beta$, no schedule adds $v$ to the independent set. The reason is that if there is a schedule that adds $v\in\beta$, then since it 4-dominates the entire graph, in the next step in the sequence we could remove all other nodes from the set, and complete it in the next step with all other $\beta$-nodes, which contradicts that assumption that there is no reconfiguration schedule.

Now, let $v\in\beta$, and consider the categories in the proof of Lemma~\ref{lemma:constant-rounds}, which shows how to add every node in $\beta$ to the independent set. Since the proof produces a schedule if there exists any node at distance 3 from $v$ or a node in $\alpha$ at distance 2, this implies that $v$ is fully connected to $\alpha$. Since this holds for any $v\in\beta$, we have that $\alpha$ and $\beta$ are fully connected, proving Condition (1).

Suppose that there exists an $\epsilon$-node $e$ that is not connected to some $a\in\alpha$ and some $b\in\beta$. Then $S_1=\{a\}$, $S_2=\{a,e\}$, $S_3=\{e\}$, $S_4=\{e,b\}$, $S_5\{b\}$ and $S_6=\beta$ is an $(\alpha,\beta,P)$-reconfiguration schedule, proving Condition (2).

By contradiction, we prove Condition (3). Suppose that there exists $T=(e_0,e_1,\ldots,e_k)$ such that $e_0\in\epsilon_\beta\setminus\epsilon_\alpha$, $e_k\in\epsilon_\alpha\setminus\epsilon_\beta$ and $\forall i<k, (e_i,e_{i+1})\not\in E$. The crux here is that we can have both $e_i$ and $e_{i+1}$ in an intermediate set. This means that every two steps, we can add an element $e_i$ of the path $T$, after removing $e_{i-2}$. Note that, as  $e_0\in\epsilon_\beta\setminus\epsilon_\alpha$ (resp. $e_k\in\epsilon_\alpha\setminus\epsilon_\beta$), there exists $a\in\alpha$ (resp. $b\in\beta$) such that $e_0$ and $a$ (resp. $e_k$ and $b$) are not connected. More formally, we have the following $(\alpha,\beta,P)$-reconfiguration schedule:
$S_1=\{a\}$, $S_2=\{a,e_0\}$, $\forall i<k, S_{2i+3}=\{e_{i}\}$, $S_{2i+4}=\{e_i,e_{i+1}\}$, $S_{2k+3}=\{e_k\}$, $S_{2k+4}=\{e_k,b\}$, $S_{2k+5}\{b\}$ and $S_{2k+6}=\beta$.

\textbf{If Conditions (1), (2), and (3) hold then there is no schedule:} Let $G=(V,E,\alpha)$ and $G=(V,E,\beta)$ be such that Conditions (1), (2), and (3) hold. Suppose that there exists an $(\alpha,\beta,P)$-reconfiguration schedule $(S_i)_{i\le N}$ of length $N$. For any $0\leq i\leq N-1$, if $S_i$ contains nodes from $\alpha$, then the only nodes that can be added in $S_{i+1}$ are nodes in $\epsilon_\beta\setminus\epsilon_\alpha$ (the other nodes are fully connected to $\alpha$, by Conditions (1) and (2)). Similarly, for any $0\leq i\leq N-1$, if $S_{i+1}$ has a node in $\beta$, then $S_i$ has only nodes from $\epsilon_\alpha\setminus\epsilon_\beta$.

Let $i_1=\max\{i~|~ S_i\cap\alpha\neq\emptyset\}$ and $i_2=\min\{i ~|~ i> i_2 \text{ and }S_i\cap\beta\neq\emptyset\}$. Since  $\alpha$ and $\beta$ are disjoint, because both sets are independent and fully connected to each other, it holds that $i_1$ and $i_2$ are well defined, and that $i_1<i_2$.

Now, $S_{i_1+1}$ only contains nodes in $\epsilon_\beta\setminus\epsilon_\alpha$, as $S_{i_1}\cap S_{i_1+1}$ must be independent. Similarly, $S_{i_2-1}$ only contains nodes in $\epsilon_\alpha\setminus\epsilon_\beta$. By definition of $i_1$ and $i_2$, for all $i_1 < i < i_2$, it holds that $S_i\in V'$. In adidtion, for all $i_1 < i < i_2$, each node in $S_{i+1}$ is connected to each node in $S_i$ in $G'$, as those nodes must form an independent set in $G$. Hence, by induction, the nodes in $\bigcup\limits_{i_1 < i < i_2}{S_i}$ form a connected component in $G'$. But this contradicts Condition (3).
\end{proof}

\section{Impossibility results for MIS reconfiguration}
\label{sec:impossibilities}
\vspace{-0.25cm}
We show here two types of impossibility results. One is the necessity of $4$-domination in the sense that $3$-domination cannot be obtained, and the other is the necessity of $\Omega(\log^*{n})$ rounds with $4$-domination, even on bounded degree graphs where it matches the complexity we provide in Corollary~\ref{cor:log-star}.

\subparagraph{Impossibility of  MIS reconfiguration with 3-domination}
\begin{theorem-repeat}{prop:3dom}
\PropThreeDom
\end{theorem-repeat}

\begin{figure}[h]
\begin{center}
\includegraphics[scale=0.75, trim= 0cm 10.7cm 0 3.7cm, clip]{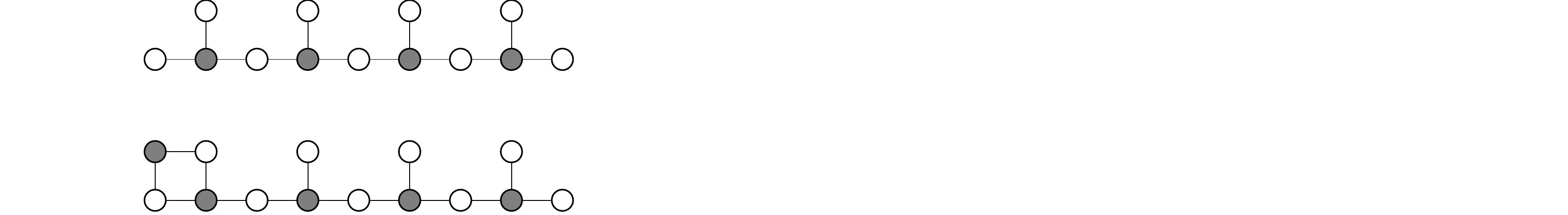}
\end{center}
\caption{White nodes are $\alpha$-nodes, and grey nodes are $\beta$-nodes. For the graph on the top, there is no schedule with 3-dominating sets. For the graph on the bottom, any schedule with 3-dominating sets must be of linear length and requires a linear number of rounds to be found.}
\label{fig:3dom}
\end{figure}

\begin{proof}[\textbf{Proof of Theorem~\ref{prop:3dom}}]\
For (1), Figure \ref{fig:3dom} (top) is an example. Consider the first $\beta$-node (grey) we try to reconfigure. This requires that all of its $\alpha$-neighbors (white) are removed from the set. But then the $\alpha$-neighbor (white) above this  $\beta$-node is not $3$-dominated by any $\alpha$-node.
For (2), Figure \ref{fig:3dom} (bottom) is an example. According to the previous argument, here we have that the top left $\beta$-node (grey) must be reconfigured first. This is possible by first removing its $\alpha$-neighbors, and then we can reconfigure the rest of the $\beta$-nodes one after the other following the path by first removing each one's $\alpha$-neighbors. This is a linear schedule which clearly requires a linear number of rounds to deduce. No shorter reconfiguration schedule is possible, by the argument for the previous item.
\end{proof}

\subparagraph{An $\Omega(\log^*n{})$  lower bound for MIS reconfiguration with 4-domination}

\begin{theorem-repeat}{prop:log-star}
\Proplogstar
\end{theorem-repeat}


\begin{figure}[h]
\begin{center}
\includegraphics[scale=0.75, trim= -1cm 14.2cm 0 5cm, clip]{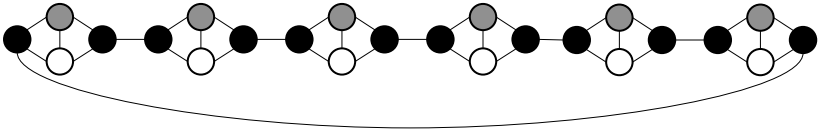}
\end{center}
\caption{White nodes are $\alpha$-nodes, grey nodes are $\beta$-nodes, and black nodes are $\epsilon$-nodes. Here, $\Omega(\log^*{n})$ rounds are needed with 4-domination.}
\label{fig:3reg}
\end{figure}

\begin{proof}[\textbf{Proof of Theorem~\ref{prop:log-star}}]
Figure \ref{fig:3reg} is an example for $k=3$. For $k\ge4$, replace the pairs of white-grey nodes by a clique of size $k-1$ containing one white and one grey node.
It is not possible to add at the same time three consecutive white nodes in the independent set, as this violates the domination.

Now, suppose we have an algorithm that finds a constant reconfiguration schedule in $o(\log^*n)$ rounds. We prove that it permits computing an MIS on a path in $o(\log^*n)$ rounds, which would be a contradiction to Linial's celebrated lower bound~\cite{Linial92}. Given the input path, we create a virtual graph where each node is replaced by two $\epsilon$-nodes connected to a clique of size $k-1$ containing an $\alpha$-node, a $\beta$-node and $k-3$ $\epsilon$-nodes (we transform the path into an instance of our counter-example). We compute a constant schedule of length $K$ in $o(\log^*n)$ rounds. Each initial node chooses the first time its $\beta$-node joined the independent set as its color in a $K$-coloring. The coloring is not necessarily proper, but it has the property that three consecutive nodes cannot have the same color. Also, as we have $K$ colors, there must be a local minimum within distance $2K$ from any node, which implies that within $O(K)=O(1)$ rounds we can obtain an MIS out of this coloring. This gives an MIS algorithm on paths in $o(\log^*n)$ rounds, which is a contradiction.
\end{proof}

\begin{corollary}
For any fixed $k$ and $d$, there exists a class of $k$-regular inputs $G=(V,E)$ with two MIS $\alpha$ and $\beta$ such that any reconfiguration schedule of constant length with $d$-domination needs $\Theta(\log^*n)$ rounds to be found.
\end{corollary}
\begin{proof}
The proof is the same as above with the same graph. The only difference is that, instead of not accepting 3 consecutive nodes to be added the same time, we now have $\left\lfloor\frac{2d+1}3\right\rfloor$ nodes.
\end{proof}

%
%
%
%
%
\section{Discussion and Open Questions}
\label{sec:discussion}
This paper defines relevant constraints for finding a reconfiguration schedule of maximal independent sets in a distributed setting. For constant-length schedules in bounded-degree graphs we completely settle the required complexity, as we provide an algorithm completing in $\Theta(\log^*n)$ communication rounds, and prove that no lower complexity exists. A main open question that remains is: Can a better complexity be found for general graphs?

Our definition only uses addition and removal of elements to the intermediate independent sets. We propose the following question: Can an efficient distributed reconfiguration schedule be found if the system allows that intermediate steps are 3-dominating and the transitions used can be any combination of addition, removal and Token Sliding?


Finally, we used as a hypothesis that the given independent sets are maximal. Our algorithm still works when the sets are not maximal, as it suffices to complete those. For example, if we are given (2,4)-ruling sets (which is equivalent to the 4-domination condition of $P$), the problem is solved. An interesting question could be to generalize for other $(a,b)$-ruling sets. What relation between $a$ and $b$ is needed to ensure that a schedule exists, and that it can be found efficiently with a distributed algorithm?


~\\\textbf{Acknowledgements:} The authors thank Alkida Balliu, Michal Dory, Seri Khoury, Dennis Olivetti, and Jukka Suomela for helpful discussions.\\
This project has received funding from the European Union's Horizon 2020 Research And Innovation Programe under grant agreement no.~755839 and it was supported in part by the Academy of Finland, Grant 285721.

\bibliographystyle{plain}
\bibliography{bib}
\end{document}